\documentclass[prl,twocolumn]{revtex4-1}
\usepackage{amssymb,latexsym}
\usepackage{amsbsy} 
\usepackage{amsmath,amsthm}
\usepackage{hyperref}
\usepackage{graphicx,color}
\usepackage{verbatim}

\newtheorem{theorem}{Theorem}

\newcommand\jpcomp[1] {j_{\mathrm{p};{#1}}}
\newcommand\jpvec {{\mathbf{j}_{\mathrm{p}}}}
\newcommand\jpvecix[1] {\mathbf{j}_{\mathrm{p};{#1}}}

\newcommand\jmvec {{\mathbf{j}_{\mathrm{m}}}}

\newcommand{\commentout}[1] {}

\begin{document}

\title{Universal lower bounds on the kinetic energy of electronic systems with noncollinear magnetism} 

\author{Erik I. Tellgren}
\email{erik.tellgren@kjemi.uio.no}
\affiliation{Hylleraas Centre for Quantum Molecular Sciences, Department of Chemistry,  University of Oslo, P.O. Box 1033 Blindern, N-0315 Oslo, Norway}

\begin{abstract}
The distribution of noncollinear magnetism in an electronic system provides information about the kinetic energy as well as some kinetic energy densities. Two different everywhere-positive kinetic densities related to the Schr\"odinger--Pauli Hamiltonian are considered. For one-electron systems described by a single Pauli spinor, the electron density, spin density and current density completely determines these kinetic energy densities. For many-electron systems, lower bounds on the kinetic energy densities are proved. These results generalize a lower bound due to von Weizs\"acker, which is based on the electron density alone and plays an important role in density functional theory. The results have applications in extensions of density functional theory that incorporate noncollinear spin densities and current densities.
\end{abstract}


\maketitle 

\section{Introduction}

Electronic states with noncollinear spin magnetization occur in molecules and materials subject to non-uniform magnetic fields, geometric frustration, or relativistic spin-orbit coupling. A global spin quantization axis is energetically unfavorable in such systems and, unlike the collinear case, the direction of local spin magnetic moments therefore varies over space~\cite{COEY_CJP65_1210}. A paradigmatic example is a spiral spin-density wave~\cite{OVERHAUSER_PR128_1437,TSUNODA_JPCM1_10427}. The distribution of noncollinear magnetism within a system in principle provides information about the underlying many-body wave function (or mixed state), and therefore also about other properties. In the present work, it is shown that the noncollinear spin magnetization density, together with the electron density and  current density, provides particularly direct information about the non-relativistic kinetic energy and kinetic energy density: For a single-electron system, the exact kinetic energy density is a simple function of these quantities. For a general many-electron system, the same simple function provides a lower bound the actual kinetic energy density. This type of bound is a general consequence of the non-relativistic quantum-mechanical description and no quantum state can give rise to densities that violate the bound.

An arbitrary $N$-electron state can be represented by a density matrix $\Gamma(\mathbf{x}_1,\ldots,\mathbf{x}_N;\mathbf{x}'_1,\ldots,\mathbf{x}'_N)$, where $\mathbf{x}_k = (\mathbf{r}_k,\omega_k)$ contains the spatial and spin coordinates. Integrating out all but one particle coordinate yields the reduced one-particle reduced density matrix $\gamma(\mathbf{x},\mathbf{x}')$, which in turn can be decomposed into natural occupation numbers (eigenvalues) $n_k$ and natural Pauli spinors (eigenvectors) $\phi_k(\mathbf{x})$,
\begin{equation}
    \gamma(\mathbf{x},\mathbf{x}') = \sum_{l=1}^{\infty} \phi_l(\mathbf{x}) \, \phi_l(\mathbf{x}')^*,
\end{equation}
where the occupation number is absorbed into the normalization, $\langle \phi_l | \phi_l \rangle = n_l$. The occupation numbers can be fractional, but always satisfy $n_l \geq 0$ and $\sum_l n_l = N$, and the decomposition is valid irrespective of whether the state $\Gamma$ is weakly or strongly correlated. Each natural spinor contributes additively to the total electron density $\rho(\mathbf{r}) = \sum_l \rho_l(\mathbf{r})$, spin density $\mathbf{m}(\mathbf{r}) = \sum_l \mathbf{m}_l(\mathbf{r})$, and paramagnetic current density $\jpvec(\mathbf{r}) = \sum_l \jpvecix{l}(\mathbf{r})$. Understanding $\phi_l(\mathbf{r}) = ( \phi_l(\mathbf{r},\uparrow), \phi_l(\mathbf{r}, \downarrow) )^T$ as a column vector containing the spin up and down component, the individual contributions are given by~\footnote{SI-based atomic units are used in this work.}
\begin{equation}
   \rho_l(\mathbf{r}) = \phi_l(\mathbf{r})^{\dagger} \phi_l(\mathbf{r}), \quad \mathbf{m}_l(\mathbf{r}) = \frac{1}{2} \phi_l(\mathbf{r})^{\dagger} \boldsymbol{\sigma} \phi_l(\mathbf{r}),
\end{equation}
where $\boldsymbol{\sigma}$ is a vector of Pauli spin matrices, and
\begin{equation}
   \jpvecix{l}(\mathbf{r}) = \mathrm{Re} \, \phi_l(\mathbf{r})^{\dagger} \mathbf{p} \phi_l(\mathbf{r}),
\end{equation}
where $\mathbf{p} = -i\nabla$ is the canonical momentum operator. Any spin densities that arise from a pure or mixed state satisfy the $N$-representability condition $\rho(\mathbf{r}) \geq 2 |\mathbf{m}(\mathbf{r})|$ everywhere in space~\cite{GONTIER_PRL111_153001}. Densities that arise from a single spinor satisfy the stronger condition $\rho = 2 |\mathbf{m}|$.

The standard, everywhere positive kinetic energy density is given $\tau(\mathbf{r}) = \sum_l \tau_l(\mathbf{r})$, where
\begin{equation}
   \tau_l(\mathbf{r}) = \frac{1}{2} \big( \mathbf{p} \phi_l(\mathbf{r}) \big)^{\dagger} \cdot \mathbf{p} \phi_l(\mathbf{r}).
\end{equation}
An alternative positive kinetic energy density, that integrates to the same total canonical kinetic energy, is $\bar{\tau}(\mathbf{r}) = \sum_l \bar{\tau}_l(\mathbf{r})$, where
\begin{equation}
   \bar{\tau}_l(\mathbf{r}) = \frac{1}{2} \big( \boldsymbol{\sigma}\cdot \mathbf{p} \phi_l(\mathbf{r}) \big)^{\dagger} \big( \boldsymbol{\sigma}\cdot \mathbf{p} \phi_l(\mathbf{r}) \big).
\end{equation}
This kinetic energy density is associated with the Pauli kinetic energy operator $(\boldsymbol{\sigma}\cdot\mathbf{p})^2/2$. The difference between the two densities is
\begin{equation}
  \bar{\tau}_l(\mathbf{r}) - \tau_l(\mathbf{r}) = \frac{i}{2} \big(\mathbf{p} \phi_l(\mathbf{r}) \big)^{\dagger} \big( \boldsymbol{\sigma}\times\mathbf{p} \phi_l(\mathbf{r}) \big).
\end{equation}

One of the earliest estimates of the kinetic energy from the particle density was established by von Weizs\"acker in the context of nuclear physics~\cite{WEIZSACKER_ZP96_431}, though it also applies to generic electronic systems and the electron density. The original bound takes the form $|\nabla\rho|^2/8\rho \leq \tau$. A more recent form that takes into account orbital magnetization and proper gauge dependence is $\tau_{\mathrm{W}}(\rho,\nabla\rho,\jpvec) = |\nabla\rho|^2/8\rho + |\jpvec|^2/2\rho \leq \tau$~\cite{BATES_JCP137_164105}, and the present work incorporates noncollinear spin densities as well. The inequality $\tau_{\mathrm{W}} \leq \tau$ is universally valid. Moreover, when the densities arise from a single spatial orbital, the bound becomes saturated, $\tau_{\mathrm{W}}(\rho,\nabla\rho,\jpvec) = \tau$. This occurs for a single electron occupying a spin-orbital with a collinear spin density and for an uncorrelated two-electron system occupying paired spin-orbitals. For densities arising from a single Pauli spinor, equality generally does not hold. Similar remarks apply to the recent work by Gontier~\cite{GONTIER_PRL111_153001} where three lower bounds on the kinetic energy were reported as part of a characterization of $N$-representable density pairs $(\rho,\mathbf{m})$. In general, these bounds remain inequalities also for densities arising from single spinors.

The von Weizs\"acker bound plays an important role in density-functional theory (DFT)~\cite{PARR89}, where it is one of the simplest meaningful density functionals. The fact that it provides a rigorous lower bound on the kinetic energy makes it a useful formal tool~\cite{LIEB_IJQC24_243}. The fact that the bound becomes an equality for single-orbital regions has made it a useful ingredient in orbital-free kinetic energy functionals~\cite{PARR89,WANG_TMCPC2000} as well as in a class of density-functional approximations (meta-GGAs) that use a kinetic energy density to detect regions of overlapping orbitals~\cite{BECKE_PRA39_3761,PERDEW_PRL82_2544,SUN_PRL111_106401}. It has also inspired quantitative measures of electron localization~\cite{BECKE_JCP92_5397,SILVI_N371_683}. To date, these concepts have been restricted to collinear spin densities.

As noncollinear magnetism is beyond the formal scope of standard DFT, the proper setting for generalizations is noncollinear spin-density-functional-theory (nc-SDFT)~\cite{KUBLER_JPF_18_469}  or current-density-functional-theory (CDFT)~\cite{VIGNALE_PRL59_2360,VIGNALE_PRB37_10685}. In nc-SDFT there is an exact, universal density functional that depends on the pair $(\rho,\mathbf{m})$ of the total density $\rho$ and spin density $\mathbf{m}$. In CDFT, the universal functional depends on either the triple $(\rho,\jpvec,\mathbf{m})$ or the pair $(\rho,\jmvec)$~\cite{CAPELLE_PRL78_1872,TELLGREN_PRA86_062506}. Here, $\jpvec$ is the paramagnetic current density and $\jmvec = \jpvec + \nabla\times\mathbf{m}$ is the magnetization current density that also includes a spin contribution. While the formulation in terms of $(\rho,\jmvec)$ couples spatial and spin degrees of freedom, and breaks the spin rotation invariance of the $(\rho,\jpvec,\mathbf{m})$ formulation, it is a natural framework when Lieb's convex analysis formulation of DFT~\cite{LIEB_IJQC24_243} is extended to CDFT~\cite{TELLGREN_PRA86_062506}. Though recent work has focused on the $\mathbf{m}$-dependence of the exchange-correlation energy~\cite{BULIK_PRB87_035117,EICH_PRL111_156401,EICH_PRB88_245102}, little work has been done on the $\mathbf{m}$- or $\jmvec$-dependence of the kinetic energy.

\section{One-electron systems}

The function $\tau_{\mathrm{W}}$ can be modified to become responsive to the spin density. A common pragmatic approach to noncollinear density functionals is to apply functionals derived for the collinear case also to noncollinear densities. To enable this, the densities $\rho_{\pm} = \tfrac{1}{2} \rho \pm |\mathbf{m}|$ are substituted for the spin up and down densities~\cite{ESCHRIG_JCC20_23,VANWULLEN_JCC23_779}. The corresponding generalization of the (spin-resolved) $\tau_{\mathrm{W}}$ is
\begin{equation}
   \tau_{\mathrm{eig}}(\rho,\jpvec,\mathbf{m}) = \frac{|\nabla \rho_{+}|^2}{8\rho_{+}} + \frac{|\nabla \rho_{-}|^2}{8\rho_{-}} + \frac{|\jpvec|^2}{2\rho}.
\end{equation}
However, this choice does not reproduce an exact kinetic energy density in the one-spinor case. Instead, define the densities
\begin{equation}
  \label{eq2CvonW}
  \tau_{\mathrm{m}}(\rho,\jpvec,\nabla \mathbf{m}^T) = \frac{|\jpvec|^2}{2\rho} + \sum_{a,b} \frac{(\nabla_a m_b)^2}{2\rho},
\end{equation}
with summation over all Cartesian components $a,b \in \{x,y,z\}$, and
\begin{equation}
    \label{eq2CvonWalt}
  \bar{\tau}_{\mathrm{m}}(\rho,\jmvec,\nabla\cdot\mathbf{m}) = \frac{|\jmvec|^2 + (\nabla\cdot \mathbf{m})^2}{2\rho}.
\end{equation}

The following theorem shows that the above functions reproduce the exact $\tau$ and $\bar{\tau}$ for arbitrary one electron systems.
\begin{theorem}[]
\label{thmTmEqK1}
Let $\rho$, $\mathbf{m}$, $\jpvec$, $\tau$, and $\bar{\tau}$ be densities arising from a single, arbitrary Pauli spinor. At every point $\mathbf{r}$ with non-zero density $\rho(\mathbf{r}) \neq 0$, it holds that (a) $\tau = \tau_{\mathrm{m}}(\rho,\jpvec,\nabla\mathbf{m}^T)$ and (b) $\bar{\tau} = \bar{\tau}_{\mathrm{m}}(\rho,\jmvec,\nabla\cdot\mathbf{m})$.
\end{theorem}
\begin{proof}
Part (a): A general Pauli spinor takes the form $\phi(\mathbf{r}) = (\lambda(\mathbf{r}), \mu(\mathbf{r}))^T$. Now fix an arbitrary point $\mathbf{r}_0$ with non-vanishing density.  Without loss of generality, the coordinate system can be chosen so that $z$-axis is aligned to the spin direction $\mathbf{m}(\mathbf{r}_0)$ at $\mathbf{r}_0$, i.e.\ $\mu(\mathbf{r}_0) = 0$. Moreover, it is easily verified that the difference $\tau - \tau_{\mathrm{m}}$ is invariant under gauge transformations. The fact $\rho(\mathbf{r}_0) \neq 0$ now implies $\lambda(\mathbf{r}_0) \neq 0$, which is sufficient to guarantee the existence of a gauge transformation $\lambda(\mathbf{r}) \mapsto \lambda(\mathbf{r}) e^{i\chi(\mathbf{r})}$, $\mu(\mathbf{r}) \mapsto \mu(\mathbf{r}) e^{i\chi(\mathbf{r})}$ that makes $\lambda(\mathbf{r})$ real-valued in a small neighborhood around $\mathbf{r}_0$. Without loss of generality, one may therefore take $\lambda(\mathbf{r}_0)$ and $\nabla \lambda(\mathbf{r}_0)$ to be real-valued. Hence, $\rho(\mathbf{r}_0) = \lambda(\mathbf{r}_0)^2$, $\jpvec(\mathbf{r}_0) = 0$, $m_x(\mathbf{r}_0) = m_y(\mathbf{r}_0) = 0$, $m_z(\mathbf{r}_0) = \tfrac{1}{2} \lambda(\mathbf{r}_0)^2$, and
\begin{equation}
  \nabla_a \mathbf{m} = \frac{1}{2} \nabla_a \phi^{\dagger} \boldsymbol{\sigma} \phi
     = \frac{1}{2}
   \begin{pmatrix}
     \lambda (\nabla_a \mu + \nabla_a \mu^*) \\
   \lambda (-i \nabla_a \mu + i \nabla_a \mu^*) \\
    2 \lambda \nabla_a \lambda
   \end{pmatrix},
\end{equation}
where $a \in \{x,y,z\}$ is a Cartesian component and all quantities are evaluated at $\mathbf{r}_0$. Insertion into the definition of $\tau_{\mathrm{m}}(\rho,\jpvec,\nabla\mathbf{m}^T)$ yields
\begin{equation}
   \tau_{\mathrm{m}} = \frac{4 \lambda^2 |\nabla \mu|^2 + 4 \lambda^2 |\nabla \lambda|^2}{8\rho} = \frac{1}{2} |\nabla \lambda|^2 + \frac{1}{2} |\nabla \mu|^2,
\end{equation}
which coincides with $\tau$. Because the point $\mathbf{r}_0$ was arbitrary, part (a) of the theorem follows.

Part (b): Under the above stipulations,
\begin{equation}
   2 \nabla\cdot\mathbf{m} = \lambda \nabla_x (\mu+\mu^*) -i \lambda \nabla_y(\mu-\mu^*) + \lambda \nabla_z \lambda
\end{equation}
and
\begin{equation}
   \jmvec = \nabla\times\mathbf{m} =
    \frac{1}{2} \lambda
   \begin{pmatrix}
       \nabla_y \lambda + i\nabla_z(\mu -  \mu^*)\\
       \nabla_z (\mu+\mu^*) - \nabla_x \lambda \\
      -i\nabla_x(\mu-\mu^*)  -\nabla_y (\mu+\mu^*)
   \end{pmatrix}
\end{equation}
and, with $\nabla_{\pm} = \nabla_x \pm i \nabla_y$,
\begin{equation}
   2\bar{\tau} = (\boldsymbol{\sigma}\cdot\mathbf{p} \phi)^{\dagger} \boldsymbol{\sigma}\cdot\mathbf{p} \phi = 
   \begin{pmatrix}
        \nabla_{-} \mu + \nabla_z \lambda \\
        \nabla_{+} \lambda  - \nabla_z \mu
   \end{pmatrix}^{\dagger}
   \begin{pmatrix}
        \nabla_{-} \mu + \nabla_z \lambda \\
        \nabla_{+} \lambda - \nabla_z \mu
   \end{pmatrix}.
\end{equation}
It is now a matter of straightforward algebra to verify that $\bar{\tau} = \bar{\tau}_{\mathrm{m}}(\rho,\jmvec,\nabla\cdot\mathbf{m})$, and thus part (b) of the theorem.
\end{proof}

\section{Many-electron systems}

For general systems, the functions $\tau_{\mathrm{m}}$ and $\bar{\tau}_{\mathrm{m}}$ are not exact, but instead provide lower bounds.
\begin{theorem}[]
\label{thmTmBound}
For densities $\rho$, $\mathbf{m}$, $\jpvec$, $\tau$, and $\bar{\tau}$ arising from arbitrary $N$-electron mixed states, (a) $\tau \geq \tau_{\mathrm{m}}(\rho,\jpvec,\mathbf{m})$ and (b) $\bar{\tau} \geq \bar{\tau}_{\mathrm{m}}(\rho,\jmvec,\nabla\cdot\mathbf{m})$.
\end{theorem}
\begin{proof}
It will first be shown that $\tau_{\mathrm{m}}$ and $\bar{\tau}_{\mathrm{m}}$ are subadditive functions. The pair $(\jpvec,\nabla\mathbf{m}^T)$ can be reorganized into a 12-dimensional vector $\mathbf{u} \in \mathbb{R}^{12}$ and the pair $(\jmvec,\nabla\cdot\mathbf{m})$ into a 4-dimensional vector $\mathbf{u} \in \mathbb{R}^4$.  Hence, both functions are instances of $f(\rho,\mathbf{u}) = |\mathbf{u}|^2/2\rho$, defined for all $\rho > 0$ and $\mathbf{u} \in \mathbb{R}^K$. Using the Young's inequality $2 \mathbf{u}\cdot\mathbf{v} \leq a |\mathbf{u}|^2 + \tfrac{1}{a} |\mathbf{v}|^2$, with $a=\sigma/\rho$, one obtains the subadditivity property
\begin{equation}
 \begin{split}
  f(\rho+\sigma,\mathbf{u}+\mathbf{v}) &  \leq \frac{(1+\sigma/\rho) |\mathbf{u}|^2 + (1+\rho/\sigma) |\mathbf{v}|^2}{2(\rho+\sigma)}
         \\
      & = f(\rho,\mathbf{u}) + f(\sigma,\mathbf{v}).
 \end{split}
\end{equation}
The fact that natural spinors contribute additively to the densities now allows iterated application of the subadditivity of $\tau_{\mathrm{m}}$ and $\bar{\tau}_{\mathrm{m}}$:
\begin{equation}
  \label{eqTAUmBOUND}
 \begin{split}
  \tau_{\mathrm{m}} \Big(\sum_j \rho_j , \sum_k \nabla^T\mathbf{m}_k, \sum_l \jpvecix{l} \Big) & \leq \sum_l \tau_{\mathrm{m}}(\rho_l, \nabla^T\mathbf{m}_l, \jpvecix{l})
         \\
      & = \sum_l \tau_l = \tau,
 \end{split}
\end{equation}
  where the identification $\tau_{\mathrm{m}}(\rho_l, \nabla^T\mathbf{m}_l, \jpvecix{l}) = \tau_l$ follows because $\tau_{\mathrm{m}}$ is exact for densities arising from a single spinor (by Theorem~\ref{thmTmEqK1}). The inequality for $\bar{\tau}_{\mathrm{m}}$ follows analogously.
\end{proof}

The original von Weizs\"acker bound, $\tau_{\mathrm{W}} \leq \tau$, and the new bound, $\tau_{\mathrm{m}} \leq \tau$, are both universally valid. However, which bound is  sharper varies with the system. In a closed-shell molecule with paired electrons, the spin density vanishes and the original bound is sharper. For a single-electron system or regions sufficiently far away from a molecule with an unpaired electron, the new bound is sharper, $\tau_{\mathrm{W}} \leq \tau_{\mathrm{m}}$. A spin density-wave with nearly uniform density is another case where $\tau_{\mathrm{m}}$ is the sharper bound.  Both bounds may be improved by taking the maximum of $\tau_{\mathrm{m}}$ and $\tau_{\mathrm{W}}$ at each point in space. Alternatively, the identity $\rho = 2 |\mathbf{m}|$ for single-electron systems provides some freedom to modify the expressions for $\tau_{\mathrm{m}}$ and $\bar{\tau}_{\mathrm{m}}$ so that they yield different estimates in the many-electron case. In general, preserving subadditivity is sufficient for preserving the lower bound properties $\tau_{\mathrm{m}} \leq \tau$ and $\bar{\tau}_{\mathrm{m}} \leq \bar{\tau}$. To this end, note Gontier's result~\cite{GONTIER_PRL111_153001} that $\tau_{\mathrm{G}} \leq 4 \tau$, where~\footnote{The notation is simplified by writing $\tau_{\mathrm{G}}(\rho,\mathbf{m})$ instead of $\tau_{\mathrm{G}}(\rho,\mathbf{m},\nabla\mathbf{m}^T)$, and similarly for $\tau_{\mathrm{mG}}$ and $\tau_{\mathrm{eig}}$.}
\begin{equation}
  \tau_{\mathrm{G}}(\rho,\mathbf{m}) = \frac{\big| \nabla \sqrt{\tfrac{1}{4} \rho^2 - |\mathbf{m}|^2} \big|^2}{2\rho}.
\end{equation}
Gontier's bound can be sharpened to $\tau_{\mathrm{G}} \leq \tau$. For densities arising from a single spinor, $\tau_{\mathrm{G}}$ vanishes identically. When instead the spin density vanishes everywhere in space, $\tau_{\mathrm{G}}$ reduces to the conventional von Weizs\"acker term $|\nabla \rho|^2/8\rho$. Moreover, $\tau_{\mathrm{G}}$ tends be large when the inequality $\tau_{\mathrm{m}} \leq \tau$ holds by a large margin. In fact, defining the sum
\begin{equation}
   \tau_{\mathrm{mG}}(\rho,\jpvec,\mathbf{m}) = \tau_{\mathrm{m}}(\rho,\jpvec,\nabla\mathbf{m}^T) + \tau_{\mathrm{G}}(\rho,\mathbf{m}),
\end{equation}
the previous inequalities can now be sharpened to $\tau_{\mathrm{mG}} \leq \tau$. While $\tau_{\mathrm{G}}$ is manifestly not subadditive, the sum $\tau_{\mathrm{mG}}$ is subadditive with respect to addition of one-spinor densities, which is sufficient to prove the sharpened inequality:
\begin{theorem}
  Let $(\rho,\jpvec,\mathbf{m})$ and $(\sigma,\mathbf{k}_{\mathrm{p}},\mathbf{n})$ be density triples from an arbitrary mixed state and a single spinor, respectively. Then
  \begin{equation}
     \tau_{\mathrm{G}}(\rho+\sigma,\mathbf{m}+\mathbf{n}) \leq \tau_{\mathrm{G}}(\rho,\mathbf{m}) + \frac{\rho \sigma |\tfrac{1}{\rho} \nabla m_c - \tfrac{1}{\sigma} \nabla n_c|^2}{2(\rho+\sigma)}
  \end{equation}
  and
  \begin{equation}
     \label{eqMGSUBADD}
    \tau_{\mathrm{mG}}(\rho+\sigma,\jpvec+\mathbf{k}_{\mathrm{p}},\mathbf{m}+\mathbf{n}) \leq \tau_{\mathrm{mG}}(\rho,\jpvec,\mathbf{m}) + \tau_{\mathrm{mG}}(\sigma,\mathbf{k}_{\mathrm{p}},\mathbf{n}).
  \end{equation}
\end{theorem}
\begin{proof}
   With the notation $f^2 = \tfrac{1}{4} \rho^2 - |\mathbf{m}|^2$ and $h^2 = \tfrac{1}{2} \rho \sigma - 2\mathbf{m}\cdot\mathbf{n}$, one can write
\begin{equation}
  \tau_{\mathrm{G}}(\rho+\sigma,\mathbf{m}+\mathbf{n}) = \frac{|\nabla f^2 + \nabla h^2|^2}{8(\rho+\sigma)(f^2+h^2)}.
\end{equation}
Using  $\nabla \rho^2 = 4 \nabla (f^2 + |\mathbf{m}|^2)$ and $\nabla \sigma^2 = 4 \nabla |\mathbf{n}|^2$ to eliminate occurences of $\nabla\rho$ and $\nabla\sigma$ when $\nabla h^2$ is written out yields
 \begin{equation}
   \nabla h^2 = \frac{\sigma}{\rho} \nabla f^2 + \frac{2}{\rho} \zeta_c \nabla m_c - \frac{2}{\sigma} \zeta_c \nabla n_c,
 \end{equation}
with implicit summation over Cartesian components $c$ and $\boldsymbol{\zeta} = \sigma \mathbf{m} - \rho \mathbf{n}$. Now temporarily assume that $\boldsymbol{\zeta}$ is non-zero; the final form below is also valid in the simpler case $\boldsymbol{\zeta} = \mathbf{0}$. Using the Young's inequality, $|\mathbf{u}+\mathbf{v}|^2 \leq (1+a) |\mathbf{u}|^2 + (1+\tfrac{1}{a}) |\mathbf{v}|^2$, $a > 0$, one now has
\begin{equation}
 \begin{split}
  & \tau_{\mathrm{G}}(\rho+\sigma, \mathbf{m}+\mathbf{n})  = \frac{\big| (1+\tfrac{\sigma}{\rho})^2 f \nabla f + \zeta_c (\tfrac{1}{\rho} \nabla m_c - \frac{1}{\sigma} \nabla n_c) \big|^2}{f^2+h^2}
              \\
     & \leq \frac{(1+a) (1+\tfrac{\sigma}{\rho})^2 f^2 |\nabla f|^2 + (1+\tfrac{1}{a}) |\zeta_c (\tfrac{1}{\rho} \nabla m_c - \frac{1}{\sigma} \nabla n_c) \big|^2}{f^2+h^2}.
 \end{split}
\end{equation}
For non-zero $\boldsymbol{\zeta}$, setting $a = |\boldsymbol{\zeta}|^2 / (\sigma (\rho+\sigma) f^2)$, noting the identity $|\boldsymbol{\zeta}|^2 = \rho \sigma h^2 - \sigma^2 f^2$, and overestimating the second term in the numerator above by replacing $\zeta_c$ with $|\boldsymbol{\zeta}|$ now yields
\begin{equation}
 \begin{split}
   \tau_{\mathrm{G}}(\rho+\sigma,\mathbf{m}+\mathbf{n}) &  \leq \frac{|\nabla f|^2}{2\rho} + \sum_c \frac{\rho \sigma \big|\tfrac{1}{\rho} \nabla m_c - \frac{1}{\sigma} \nabla n_c\big|^2}{2(\rho+\sigma)}.
 \end{split}
\end{equation}
The first term on the right-hand side can be identified with $\tau_{\mathrm{G}}(\rho,\mathbf{m})$, proving the first part of the theorem.

The second part of the theorem follows because cross terms involving $\nabla m_c \cdot \nabla n_c$ from $\tau_{\mathrm{G}}$ and $\tau_{\mathrm{m}}$ exactly cancel. Hence, adding $\tau_{\mathrm{m}}$ to the above inequality yields
\begin{equation}
 \begin{split}
  \tau_{\mathrm{mG}}(\rho+\sigma &, \jpvec+\mathbf{k}_{\mathrm{p}},\mathbf{m}+\mathbf{n}) \leq  \tau_{\mathrm{G}}(\rho,\mathbf{m}) + \frac{|\jpvec+\mathbf{k}_{\mathrm{p}}|^2}{2(\rho+\sigma)}
          \\
    &  + \sum_c \frac{(1+\tfrac{\sigma}{\rho})|\nabla m_c|^2 + (1+\tfrac{\rho}{\sigma}) |\nabla n_c|^2}{2(\rho+\sigma)}.
 \end{split}
\end{equation}
Identifying spin terms and noting the subadditivity of the current density term, $|\jpvec+\mathbf{k}_{\mathrm{p}}|^2/(2\rho+2\sigma) \leq |\jpvec|^2/2\rho + |\mathbf{k}_{\mathrm{p}}|^2/2\sigma$, one obtains
\begin{equation}
  \tau_{\mathrm{mG}}(\rho+\sigma,\jpvec+\mathbf{k}_{\mathrm{p}},\mathbf{m}+\mathbf{n}) \leq  \tau_{\mathrm{mG}}(\rho,\jpvec,\mathbf{m}) + \tau_{\mathrm{m}}(\sigma,\mathbf{k}_{\mathrm{p}},\mathbf{n}).
\end{equation}
Finally, $\tau_{\mathrm{G}}(\sigma,\mathbf{n})$ is identically zero and can be added to produce to the form in Eq.~\eqref{eqMGSUBADD}, completing the proof.
\end{proof}

\begin{theorem}[]
\label{thmTmGBound}
For densities $\rho$, $\mathbf{m}$, $\jpvec$, $\tau$ arising from an arbitrary mixed state, $\tau \geq \tau_{\mathrm{mG}}(\rho,\jpvec,\mathbf{m})$.
\end{theorem}
\begin{proof}
Decompose the densities into contributions from individual spinors and iterate the subaddivity result in Eq.~\eqref{eqMGSUBADD},
 \begin{equation}
    \tau_{\mathrm{mG}}(\rho,\mathbf{m},\jpvec) \leq \sum_l \tau_{\mathrm{mG}}(\rho_l,\mathbf{m}_l,\jpvecix{l})  = \sum_l \tau_l = \tau.
 \end{equation}
\end{proof}

For collinear spin densities, both $\tau_{\mathrm{eig}}$ and $\tau_{\mathrm{mG}}$ reduce to a known form of the von Weizs\"acker bound. For noncollinear spin densities, the following theorem shows that $\tau_{\mathrm{mG}}$ is the sharper bound.
\begin{theorem}
  For a density triple $(\rho,\jpvec,\mathbf{m})$ arising from an arbitrary mixed state, $\tau_{\mathrm{eig}}(\rho,\jpvec,\mathbf{m}) \leq \tau_{\mathrm{mG}}(\rho,\jpvec,\mathbf{m})$.
\end{theorem}
\begin{proof}
 The alternative expression
\begin{equation}
   \tau_{\mathrm{eig}}(\rho,\jpvec,\mathbf{m}) = \frac{|\jpvec|^2}{2\rho} + \frac{\big|\tfrac{1}{2} \nabla (\rho_{+} - \rho_{-}) \big|^2}{2\rho} + \frac{\big|\nabla \sqrt{\rho_{+} \rho_{-}} \big|^2}{2\rho}.
\end{equation} 
can be verified by a direct calculation. The last term above can be identified as $\tau_{\mathrm{G}}(\rho,\mathbf{m}) = \big|\nabla \sqrt{\rho_{+} \rho_{-}} \big|^2/2\rho$. Then, iterating the general gradient inequality $|\nabla\sqrt{f^2+g^2}|^2 \leq |\nabla f|^2 +|\nabla g|^2$~\cite[Sec.~ 6.17, 7.8]{LIEB01},
\begin{equation}
   \big|\tfrac{1}{2} \nabla (\rho_{+} - \rho_{-}) \big|^2 = \big| \nabla |\mathbf{m}| \big|^2 \leq \sum_c |\nabla m_c|^2,
\end{equation}
where the right-hand side coincides with the spin-density terms in the numerator of $\tau_{\mathrm{m}}$. The theorem follows.
\end{proof}

\section{Discussion}

The above results provide a foundation for generalization of common measures of electron localization, orbital overlap, and exchange hole curvature~\cite{BECKE_PRA39_3761,BECKE_JCP92_5397,SILVI_N371_683,DOBSON_JCP94_4328,FURNESS_MP114_1415} to a noncollinear setting. The {\it isoorbital indicator} $\alpha_{\mathrm{W}} = (\tau - \tau_{\mathrm{W}}) / \tau_{\mathrm{unif}} \geq 0$, where $\tau_{\mathrm{unif}}$ is an estimate of the kinetic energy density in a uniform electron gas, has been advocated as the proper measure of orbital overlap in meta-GGA functionals~\cite{SUN_PRL111_106401}. Theorems~\ref{thmTmEqK1}, \ref{thmTmBound}, and \ref{thmTmGBound} show that one can define analogous {\it isospinor indicators}, e.g.\ $\alpha_{\mathrm{mG}} = (\tau - \tau_{\mathrm{mG}}) / \tau_{\mathrm{unif}} \geq 0$ and $\bar{\alpha}_{\mathrm{m}} = (\bar{\tau} - \bar{\tau}_{\mathrm{m}}) / \tau_{\mathrm{unif}} \geq 0$ are able to discriminate between single- and many-spinor regions.

As in the collinear case, the quantities appearing in the lower bounds are related to the curvature of the pair density. For an uncorrelated state with natural spinors $\phi_k$, the pair density is a sum  $n(\mathbf{r},\mathbf{r}') = \rho(\mathbf{r}) \rho(\mathbf{r}') + n_{\mathrm{x}}(\mathbf{r},\mathbf{r}')$ of a direct density product and the exchange term $n_{\mathrm{x}}(\mathbf{r},\mathbf{r}') = -\sum_{kl} \phi_l(\mathbf{r}')^{\dagger} \phi_k(\mathbf{r}) \phi_k(\mathbf{r})^{\dagger} \phi_l(\mathbf{r}')$. Many density functional approximations rely on modelling of the exchange hole, defined as $h_{\mathrm{x}}(\mathbf{r},\mathbf{r}') = n_{\mathrm{x}}(\mathbf{r},\mathbf{r}') / \rho(\mathbf{r}')$~\cite{PARR89}. Some meta-GGA models rely specifically on the fact that the curvature at coinciding electron locations,
\begin{equation}
  n''_{\mathrm{x}}(\mathbf{r}) = \nabla^2 n_{\mathrm{x}}(\mathbf{r},\mathbf{r}') |_{\mathbf{r}'=\mathbf{r}},
\end{equation}
is related to $\tau-\tau_{\mathrm{W}}$ in collinear systems~\cite{BECKE_PRA39_3761,DOBSON_JCP94_4328}. With the notation $K_{ca} =  J_{ca} - i \nabla_a m_c = \sum_l \phi_l^{\dagger} \sigma_c p_a \phi_l$, where $J_{ca}$ is real, and $\tau^c = \tfrac{1}{2} \sum_l (\nabla \phi_l)^{\dagger} \cdot \sigma_c \nabla \phi_l$, the exchange density curvature is
\begin{equation}
  n''_{\mathrm{x}} = 2\rho(\tau - \tau_{\mathrm{W}}) + 4m_c \tau^c - K_{ca}^* K_{ca} - \frac{1}{2} \rho \nabla^2 \rho - 2 m_c \nabla^2 m_c.
\end{equation}
For any location in space, one may align the coordinate system so that $\mathbf{m}$ is parallel to the $z$-axis.  Expanding the inequality $\psi^{\dagger} \psi \geq 0$, with $\psi = (1-\sigma_c) (p_a + A_a) \phi$ and $A_a = -\jpcomp{a}/\rho$, yields $4m_c \tau^c - J_{ca}^2 \leq 4|\mathbf{m}| (\tau - |\jpvec|^2/2\rho)$. Thus,
\begin{equation}
 \begin{split}
  n''_{\mathrm{x}} & \leq 2(\rho+2|\mathbf{m}|) \Big( \tau - \frac{|\jpvec|^2}{2\rho} - \frac{|\nabla\rho|^2 + 4|\nabla_a m_b|^2}{8(\rho+2|\mathbf{m}|)} \Big)
             \\
      & \ \ \  - \frac{1}{2} \rho \nabla^2 \rho - 2 m_c \nabla^2 m_c.
 \end{split}
\end{equation}
Hence, the exchange-hole curvature involves the sum of $\tau_{\mathrm{W}}$ and $\tau_{\mathrm{m}}$, corrected by a $\jpvec$-dependent term that yields the correct gauge dependence.

The functions $\tau_{\mathrm{mG}}$ and $\bar{\tau}_{\mathrm{m}}$ also provide the foundation for a new type of orbital-free kinetic energy functionals, $T_{\mathrm{mG}}[\rho,\jpvec,\mathbf{m}] = \int \tau_{\mathrm{mG}}(\rho(\mathbf{r}),\jpvec(\mathbf{r}),\mathbf{m}(\mathbf{r})) \, d\mathbf{r}$ and $\bar{T}_{\mathrm{m}}[\rho,\jmvec,\nabla\cdot\mathbf{m}] = \int \bar{\tau}_{\mathrm{mG}}(\rho(\mathbf{r}),\jmvec(\mathbf{r}),\nabla\cdot\mathbf{m}(\mathbf{r})) \, d\mathbf{r}$, which incorporate noncollinear spin densities and are rigorous lower bounds on the true kinetic energy.  In order to stay within the strict $(\rho,\jmvec)$-formulation of CDFT, the divergence term must either be expressed as a functional of $(\rho,\jmvec)$, which is likely to be complicated, or simply omitted. Omitting the term preserves the lower bound property, while exactness in the single-spinor case is lost and the term is generally not small.

Finally, the fact that the results presented here bound not only the kinetic energy, but apply pointwise to the kinetic energy density, facilitates numerical exploration. The bounds apply to any densities obtained from numerical electronic-structure structure methods. Additionally, the bounds apply when $\phi_l(\mathbf{r})$ and $\nabla\phi_l(\mathbf{r})$ are treated as random variables. Orthonormality constrains only integrals over all space, not pointwise values. Hence, adopting a statistical model and generating pointwise densities from randomized spinors is a possible direction for obtaining numerical quantification of the gaps between the exact kinetic energies and the lower bounds.  Modifications of the bounds that improve their statistical accuracy can also be explored.

\section*{Acknowledgments}

This work was supported by the Norwegian Research Council through the Grant No.~240674 and CoE Centre for Theoretical and Computational Chemistry (CTCC) Grant No. 179568/V30. The author thanks T.~Helgaker, A.~Borgoo, A.~Laestadius, and S.~Kvaal for useful discussions.


\begin{thebibliography}{29}%
\makeatletter
\providecommand \@ifxundefined [1]{%
 \@ifx{#1\undefined}
}%
\providecommand \@ifnum [1]{%
 \ifnum #1\expandafter \@firstoftwo
 \else \expandafter \@secondoftwo
 \fi
}%
\providecommand \@ifx [1]{%
 \ifx #1\expandafter \@firstoftwo
 \else \expandafter \@secondoftwo
 \fi
}%
\providecommand \natexlab [1]{#1}%
\providecommand \enquote  [1]{``#1''}%
\providecommand \bibnamefont  [1]{#1}%
\providecommand \bibfnamefont [1]{#1}%
\providecommand \citenamefont [1]{#1}%
\providecommand \href@noop [0]{\@secondoftwo}%
\providecommand \href [0]{\begingroup \@sanitize@url \@href}%
\providecommand \@href[1]{\@@startlink{#1}\@@href}%
\providecommand \@@href[1]{\endgroup#1\@@endlink}%
\providecommand \@sanitize@url [0]{\catcode `\\12\catcode `\$12\catcode
  `\&12\catcode `\#12\catcode `\^12\catcode `\_12\catcode `\%12\relax}%
\providecommand \@@startlink[1]{}%
\providecommand \@@endlink[0]{}%
\providecommand \url  [0]{\begingroup\@sanitize@url \@url }%
\providecommand \@url [1]{\endgroup\@href {#1}{\urlprefix }}%
\providecommand \urlprefix  [0]{URL }%
\providecommand \Eprint [0]{\href }%
\providecommand \doibase [0]{http://dx.doi.org/}%
\providecommand \selectlanguage [0]{\@gobble}%
\providecommand \bibinfo  [0]{\@secondoftwo}%
\providecommand \bibfield  [0]{\@secondoftwo}%
\providecommand \translation [1]{[#1]}%
\providecommand \BibitemOpen [0]{}%
\providecommand \bibitemStop [0]{}%
\providecommand \bibitemNoStop [0]{.\EOS\space}%
\providecommand \EOS [0]{\spacefactor3000\relax}%
\providecommand \BibitemShut  [1]{\csname bibitem#1\endcsname}%
\let\auto@bib@innerbib\@empty
\bibitem [{\citenamefont {Coey}(1987)}]{COEY_CJP65_1210}%
  \BibitemOpen
  \bibfield  {author} {\bibinfo {author} {\bibfnamefont {J.~M.~D.}\
  \bibnamefont {Coey}},\ }\href@noop {} {\bibfield  {journal} {\bibinfo
  {journal} {Can. J. Phys.}\ }\textbf {\bibinfo {volume} {65}},\ \bibinfo
  {pages} {1210} (\bibinfo {year} {1987})}\BibitemShut {NoStop}%
\bibitem [{\citenamefont {Overhauser}(1962)}]{OVERHAUSER_PR128_1437}%
  \BibitemOpen
  \bibfield  {author} {\bibinfo {author} {\bibfnamefont {A.~W.}\ \bibnamefont
  {Overhauser}},\ }\href@noop {} {\bibfield  {journal} {\bibinfo  {journal}
  {Phys. Rev.}\ }\textbf {\bibinfo {volume} {128}},\ \bibinfo {pages} {1437}
  (\bibinfo {year} {1962})}\BibitemShut {NoStop}%
\bibitem [{\citenamefont {Tsunoda}(1989)}]{TSUNODA_JPCM1_10427}%
  \BibitemOpen
  \bibfield  {author} {\bibinfo {author} {\bibfnamefont {Y.}~\bibnamefont
  {Tsunoda}},\ }\href@noop {} {\bibfield  {journal} {\bibinfo  {journal} {J.
  Phys.: Condens. Matter}\ }\textbf {\bibinfo {volume} {1}},\ \bibinfo {pages}
  {10427} (\bibinfo {year} {1989})}\BibitemShut {NoStop}%
\bibitem [{Note1()}]{Note1}%
  \BibitemOpen
  \bibinfo {note} {SI-based atomic units are used in this work.}\BibitemShut
  {Stop}%
\bibitem [{\citenamefont {Gontier}(2013)}]{GONTIER_PRL111_153001}%
  \BibitemOpen
  \bibfield  {author} {\bibinfo {author} {\bibfnamefont {D.}~\bibnamefont
  {Gontier}},\ }\href {\doibase 10.1103/PhysRevLett.111.153001} {\bibfield
  {journal} {\bibinfo  {journal} {Phys. Rev. Lett.}\ }\textbf {\bibinfo
  {volume} {111}},\ \bibinfo {pages} {153001} (\bibinfo {year}
  {2013})}\BibitemShut {NoStop}%
\bibitem [{\citenamefont {von Weizs\"acker}(1935)}]{WEIZSACKER_ZP96_431}%
  \BibitemOpen
  \bibfield  {author} {\bibinfo {author} {\bibfnamefont {C.~F.}\ \bibnamefont
  {von Weizs\"acker}},\ }\href@noop {} {\bibfield  {journal} {\bibinfo
  {journal} {Z. Phys.}\ }\textbf {\bibinfo {volume} {96}},\ \bibinfo {pages}
  {431} (\bibinfo {year} {1935})}\BibitemShut {NoStop}%
\bibitem [{\citenamefont {Bates}\ and\ \citenamefont
  {Furche}(2012)}]{BATES_JCP137_164105}%
  \BibitemOpen
  \bibfield  {author} {\bibinfo {author} {\bibfnamefont {J.~E.}\ \bibnamefont
  {Bates}}\ and\ \bibinfo {author} {\bibfnamefont {F.}~\bibnamefont {Furche}},\
  }\href@noop {} {\bibfield  {journal} {\bibinfo  {journal} {J. Chem. Phys.}\
  }\textbf {\bibinfo {volume} {137}},\ \bibinfo {pages} {164105} (\bibinfo
  {year} {2012})}\BibitemShut {NoStop}%
\bibitem [{\citenamefont {Parr}\ and\ \citenamefont {Yang}(1989)}]{PARR89}%
  \BibitemOpen
  \bibfield  {author} {\bibinfo {author} {\bibfnamefont {R.~G.}\ \bibnamefont
  {Parr}}\ and\ \bibinfo {author} {\bibfnamefont {W.}~\bibnamefont {Yang}},\
  }\href@noop {} {\emph {\bibinfo {title} {Density-Functional Theory of Atoms
  and Molecules}}}\ (\bibinfo  {publisher} {Oxford University Press},\ \bibinfo
  {year} {1989})\BibitemShut {NoStop}%
\bibitem [{\citenamefont {Lieb}(1983)}]{LIEB_IJQC24_243}%
  \BibitemOpen
  \bibfield  {author} {\bibinfo {author} {\bibfnamefont {E.~H.}\ \bibnamefont
  {Lieb}},\ }\href@noop {} {\bibfield  {journal} {\bibinfo  {journal} {Int. J.
  Quantum Chem.}\ }\textbf {\bibinfo {volume} {24}},\ \bibinfo {pages} {243}
  (\bibinfo {year} {1983})}\BibitemShut {NoStop}%
\bibitem [{\citenamefont {Wang}\ and\ \citenamefont
  {Carter}(2000)}]{WANG_TMCPC2000}%
  \BibitemOpen
  \bibfield  {author} {\bibinfo {author} {\bibfnamefont {Y.~A.}\ \bibnamefont
  {Wang}}\ and\ \bibinfo {author} {\bibfnamefont {E.~A.}\ \bibnamefont
  {Carter}},\ }in\ \href@noop {} {\emph {\bibinfo {booktitle} {Theoretical
  Methods in Condensed Phase Chemistry}}},\ \bibinfo {series and number}
  {Progress in Theoretical Chemistry and Physics},\ \bibinfo {editor} {edited
  by\ \bibinfo {editor} {\bibfnamefont {S.~D.}\ \bibnamefont {Schwartz}}}\
  (\bibinfo  {publisher} {Kluwer},\ \bibinfo {address} {Dordrecht},\ \bibinfo
  {year} {2000})\ pp.\ \bibinfo {pages} {117--184}\BibitemShut {NoStop}%
\bibitem [{\citenamefont {Becke}\ and\ \citenamefont
  {Roussel}(1989)}]{BECKE_PRA39_3761}%
  \BibitemOpen
  \bibfield  {author} {\bibinfo {author} {\bibfnamefont {A.~D.}\ \bibnamefont
  {Becke}}\ and\ \bibinfo {author} {\bibfnamefont {M.~R.}\ \bibnamefont
  {Roussel}},\ }\href {\doibase 10.1103/PhysRevA.39.3761} {\bibfield  {journal}
  {\bibinfo  {journal} {Phys. Rev. A}\ }\textbf {\bibinfo {volume} {39}},\
  \bibinfo {pages} {3761} (\bibinfo {year} {1989})}\BibitemShut {NoStop}%
\bibitem [{\citenamefont {Perdew}\ \emph {et~al.}(1999)\citenamefont {Perdew},
  \citenamefont {Kurth}, \citenamefont {Zupan},\ and\ \citenamefont
  {Blaha}}]{PERDEW_PRL82_2544}%
  \BibitemOpen
  \bibfield  {author} {\bibinfo {author} {\bibfnamefont {J.~P.}\ \bibnamefont
  {Perdew}}, \bibinfo {author} {\bibfnamefont {S.}~\bibnamefont {Kurth}},
  \bibinfo {author} {\bibfnamefont {A.}~\bibnamefont {Zupan}}, \ and\ \bibinfo
  {author} {\bibfnamefont {P.}~\bibnamefont {Blaha}},\ }\href {\doibase
  10.1103/PhysRevLett.82.2544} {\bibfield  {journal} {\bibinfo  {journal}
  {Phys. Rev. Lett.}\ }\textbf {\bibinfo {volume} {82}},\ \bibinfo {pages}
  {2544} (\bibinfo {year} {1999})}\BibitemShut {NoStop}%
\bibitem [{\citenamefont {Sun}\ \emph {et~al.}(2013)\citenamefont {Sun},
  \citenamefont {Xiao}, \citenamefont {Fang}, \citenamefont {Haunschild},
  \citenamefont {Hao}, \citenamefont {Ruzsinszky}, \citenamefont {Csonka},
  \citenamefont {Scuseria},\ and\ \citenamefont {Perdew}}]{SUN_PRL111_106401}%
  \BibitemOpen
  \bibfield  {author} {\bibinfo {author} {\bibfnamefont {J.}~\bibnamefont
  {Sun}}, \bibinfo {author} {\bibfnamefont {B.}~\bibnamefont {Xiao}}, \bibinfo
  {author} {\bibfnamefont {Y.}~\bibnamefont {Fang}}, \bibinfo {author}
  {\bibfnamefont {R.}~\bibnamefont {Haunschild}}, \bibinfo {author}
  {\bibfnamefont {P.}~\bibnamefont {Hao}}, \bibinfo {author} {\bibfnamefont
  {A.}~\bibnamefont {Ruzsinszky}}, \bibinfo {author} {\bibfnamefont {G.~I.}\
  \bibnamefont {Csonka}}, \bibinfo {author} {\bibfnamefont {G.~E.}\
  \bibnamefont {Scuseria}}, \ and\ \bibinfo {author} {\bibfnamefont {J.~P.}\
  \bibnamefont {Perdew}},\ }\href {\doibase 10.1103/PhysRevLett.111.106401}
  {\bibfield  {journal} {\bibinfo  {journal} {Phys. Rev. Lett.}\ }\textbf
  {\bibinfo {volume} {111}},\ \bibinfo {pages} {106401} (\bibinfo {year}
  {2013})}\BibitemShut {NoStop}%
\bibitem [{\citenamefont {Becke}\ and\ \citenamefont
  {Edgecombe}(1990)}]{BECKE_JCP92_5397}%
  \BibitemOpen
  \bibfield  {author} {\bibinfo {author} {\bibfnamefont {A.~D.}\ \bibnamefont
  {Becke}}\ and\ \bibinfo {author} {\bibfnamefont {K.~E.}\ \bibnamefont
  {Edgecombe}},\ }\href@noop {} {\bibfield  {journal} {\bibinfo  {journal} {J.
  Chem. Phys.}\ }\textbf {\bibinfo {volume} {92}},\ \bibinfo {pages} {5397}
  (\bibinfo {year} {1990})}\BibitemShut {NoStop}%
\bibitem [{\citenamefont {Silvi}\ and\ \citenamefont
  {Savin}(1994)}]{SILVI_N371_683}%
  \BibitemOpen
  \bibfield  {author} {\bibinfo {author} {\bibfnamefont {B.}~\bibnamefont
  {Silvi}}\ and\ \bibinfo {author} {\bibfnamefont {A.}~\bibnamefont {Savin}},\
  }\href@noop {} {\bibfield  {journal} {\bibinfo  {journal} {Nature}\ }\textbf
  {\bibinfo {volume} {371}},\ \bibinfo {pages} {683} (\bibinfo {year}
  {1994})}\BibitemShut {NoStop}%
\bibitem [{\citenamefont {K{\"u}bler}\ \emph {et~al.}(1988)\citenamefont
  {K{\"u}bler}, \citenamefont {H{\"o}ck}, \citenamefont {Sticht},\ and\
  \citenamefont {Williams}}]{KUBLER_JPF_18_469}%
  \BibitemOpen
  \bibfield  {author} {\bibinfo {author} {\bibfnamefont {J.}~\bibnamefont
  {K{\"u}bler}}, \bibinfo {author} {\bibfnamefont {K.-H.}\ \bibnamefont
  {H{\"o}ck}}, \bibinfo {author} {\bibfnamefont {J.}~\bibnamefont {Sticht}}, \
  and\ \bibinfo {author} {\bibfnamefont {A.~R.}\ \bibnamefont {Williams}},\
  }\href@noop {} {\bibfield  {journal} {\bibinfo  {journal} {J. Phys. F: Met.
  Phys.}\ }\textbf {\bibinfo {volume} {18}},\ \bibinfo {pages} {469} (\bibinfo
  {year} {1988})}\BibitemShut {NoStop}%
\bibitem [{\citenamefont {Vignale}\ and\ \citenamefont
  {Rasolt}(1987)}]{VIGNALE_PRL59_2360}%
  \BibitemOpen
  \bibfield  {author} {\bibinfo {author} {\bibfnamefont {G.}~\bibnamefont
  {Vignale}}\ and\ \bibinfo {author} {\bibfnamefont {M.}~\bibnamefont
  {Rasolt}},\ }\href {\doibase 10.1103/PhysRevLett.59.2360} {\bibfield
  {journal} {\bibinfo  {journal} {Phys. Rev. Lett.}\ }\textbf {\bibinfo
  {volume} {59}},\ \bibinfo {pages} {2360} (\bibinfo {year}
  {1987})}\BibitemShut {NoStop}%
\bibitem [{\citenamefont {Vignale}\ and\ \citenamefont
  {Rasolt}(1988)}]{VIGNALE_PRB37_10685}%
  \BibitemOpen
  \bibfield  {author} {\bibinfo {author} {\bibfnamefont {G.}~\bibnamefont
  {Vignale}}\ and\ \bibinfo {author} {\bibfnamefont {M.}~\bibnamefont
  {Rasolt}},\ }\href {\doibase 10.1103/PhysRevB.37.10685} {\bibfield  {journal}
  {\bibinfo  {journal} {Phys. Rev. B}\ }\textbf {\bibinfo {volume} {37}},\
  \bibinfo {pages} {10685} (\bibinfo {year} {1988})}\BibitemShut {NoStop}%
\bibitem [{\citenamefont {Capelle}\ and\ \citenamefont
  {Gross}(1997)}]{CAPELLE_PRL78_1872}%
  \BibitemOpen
  \bibfield  {author} {\bibinfo {author} {\bibfnamefont {K.}~\bibnamefont
  {Capelle}}\ and\ \bibinfo {author} {\bibfnamefont {E.~K.~U.}\ \bibnamefont
  {Gross}},\ }\href@noop {} {\bibfield  {journal} {\bibinfo  {journal} {Phys.
  Rev. Lett.}\ }\textbf {\bibinfo {volume} {78}},\ \bibinfo {pages} {1872}
  (\bibinfo {year} {1997})}\BibitemShut {NoStop}%
\bibitem [{\citenamefont {Tellgren}\ \emph {et~al.}(2012)\citenamefont
  {Tellgren}, \citenamefont {Kvaal}, \citenamefont {Sagvolden}, \citenamefont
  {Ekstr{\"o}m}, \citenamefont {Teale},\ and\ \citenamefont
  {Helgaker}}]{TELLGREN_PRA86_062506}%
  \BibitemOpen
  \bibfield  {author} {\bibinfo {author} {\bibfnamefont {E.~I.}\ \bibnamefont
  {Tellgren}}, \bibinfo {author} {\bibfnamefont {S.}~\bibnamefont {Kvaal}},
  \bibinfo {author} {\bibfnamefont {E.}~\bibnamefont {Sagvolden}}, \bibinfo
  {author} {\bibfnamefont {U.}~\bibnamefont {Ekstr{\"o}m}}, \bibinfo {author}
  {\bibfnamefont {A.~M.}\ \bibnamefont {Teale}}, \ and\ \bibinfo {author}
  {\bibfnamefont {T.}~\bibnamefont {Helgaker}},\ }\href@noop {} {\bibfield
  {journal} {\bibinfo  {journal} {Phys. Rev. A}\ }\textbf {\bibinfo {volume}
  {86}},\ \bibinfo {pages} {062506} (\bibinfo {year} {2012})}\BibitemShut
  {NoStop}%
\bibitem [{\citenamefont {Bulik}\ \emph {et~al.}(2013)\citenamefont {Bulik},
  \citenamefont {Scalmani}, \citenamefont {Frisch},\ and\ \citenamefont
  {Scuseria}}]{BULIK_PRB87_035117}%
  \BibitemOpen
  \bibfield  {author} {\bibinfo {author} {\bibfnamefont {I.~W.}\ \bibnamefont
  {Bulik}}, \bibinfo {author} {\bibfnamefont {G.}~\bibnamefont {Scalmani}},
  \bibinfo {author} {\bibfnamefont {M.~J.}\ \bibnamefont {Frisch}}, \ and\
  \bibinfo {author} {\bibfnamefont {G.~E.}\ \bibnamefont {Scuseria}},\
  }\href@noop {} {\bibfield  {journal} {\bibinfo  {journal} {Phys. Rev. B}\
  }\textbf {\bibinfo {volume} {87}},\ \bibinfo {pages} {035117} (\bibinfo
  {year} {2013})}\BibitemShut {NoStop}%
\bibitem [{\citenamefont {Eich}\ and\ \citenamefont
  {Gross}(2013)}]{EICH_PRL111_156401}%
  \BibitemOpen
  \bibfield  {author} {\bibinfo {author} {\bibfnamefont {F.~G.}\ \bibnamefont
  {Eich}}\ and\ \bibinfo {author} {\bibfnamefont {E.~K.~U.}\ \bibnamefont
  {Gross}},\ }\href@noop {} {\bibfield  {journal} {\bibinfo  {journal} {Phys.
  Rev. Lett.}\ }\textbf {\bibinfo {volume} {111}},\ \bibinfo {pages} {156401}
  (\bibinfo {year} {2013})}\BibitemShut {NoStop}%
\bibitem [{\citenamefont {Eich}\ \emph {et~al.}(2013)\citenamefont {Eich},
  \citenamefont {Pittalis},\ and\ \citenamefont {Vignale}}]{EICH_PRB88_245102}%
  \BibitemOpen
  \bibfield  {author} {\bibinfo {author} {\bibfnamefont {F.~G.}\ \bibnamefont
  {Eich}}, \bibinfo {author} {\bibfnamefont {S.}~\bibnamefont {Pittalis}}, \
  and\ \bibinfo {author} {\bibfnamefont {G.}~\bibnamefont {Vignale}},\
  }\href@noop {} {\bibfield  {journal} {\bibinfo  {journal} {Phys. Rev. B}\
  }\textbf {\bibinfo {volume} {88}},\ \bibinfo {pages} {245102} (\bibinfo
  {year} {2013})}\BibitemShut {NoStop}%
\bibitem [{\citenamefont {Eschrig}\ and\ \citenamefont
  {Servedio}(1999)}]{ESCHRIG_JCC20_23}%
  \BibitemOpen
  \bibfield  {author} {\bibinfo {author} {\bibfnamefont {H.}~\bibnamefont
  {Eschrig}}\ and\ \bibinfo {author} {\bibfnamefont {V.~D.~P.}\ \bibnamefont
  {Servedio}},\ }\href@noop {} {\bibfield  {journal} {\bibinfo  {journal} {J.
  Comput. Chem.}\ }\textbf {\bibinfo {volume} {20}},\ \bibinfo {pages} {23}
  (\bibinfo {year} {1999})}\BibitemShut {NoStop}%
\bibitem [{\citenamefont {van W{\"u}llen}(2002)}]{VANWULLEN_JCC23_779}%
  \BibitemOpen
  \bibfield  {author} {\bibinfo {author} {\bibfnamefont {C.}~\bibnamefont {van
  W{\"u}llen}},\ }\href@noop {} {\bibfield  {journal} {\bibinfo  {journal} {J.
  Comp. Chem.}\ }\textbf {\bibinfo {volume} {23}},\ \bibinfo {pages} {779}
  (\bibinfo {year} {2002})}\BibitemShut {NoStop}%
\bibitem [{Note2()}]{Note2}%
  \BibitemOpen
  \bibinfo {note} {The notation is simplified by writing $\tau _{\protect
  \mathrm {G}}(\rho ,\protect \mathbf {m})$ instead of $\tau _{\protect \mathrm
  {G}}(\rho ,\protect \mathbf {m},\nabla \protect \mathbf {m}^T)$, and
  similarly for $\tau _{\protect \mathrm {mG}}$ and $\tau _{\protect \mathrm
  {eig}}$.}\BibitemShut {Stop}%
\bibitem [{\citenamefont {Lieb}\ and\ \citenamefont {Loss}(2001)}]{LIEB01}%
  \BibitemOpen
  \bibfield  {author} {\bibinfo {author} {\bibfnamefont {E.~H.}\ \bibnamefont
  {Lieb}}\ and\ \bibinfo {author} {\bibfnamefont {M.}~\bibnamefont {Loss}},\
  }\href@noop {} {\emph {\bibinfo {title} {Analysis}}}\ (\bibinfo  {publisher}
  {American Mathematical Society},\ \bibinfo {year} {2001})\BibitemShut
  {NoStop}%
\bibitem [{\citenamefont {Dobson}(1991)}]{DOBSON_JCP94_4328}%
  \BibitemOpen
  \bibfield  {author} {\bibinfo {author} {\bibfnamefont {J.~F.}\ \bibnamefont
  {Dobson}},\ }\href@noop {} {\bibfield  {journal} {\bibinfo  {journal} {J.
  Chem. Phys.}\ }\textbf {\bibinfo {volume} {94}},\ \bibinfo {pages} {4328}
  (\bibinfo {year} {1991})}\BibitemShut {NoStop}%
\bibitem [{\citenamefont {Furness}\ \emph {et~al.}(2016)\citenamefont
  {Furness}, \citenamefont {Ekstr{\"o}m}, \citenamefont {Helgaker},\ and\
  \citenamefont {Teale}}]{FURNESS_MP114_1415}%
  \BibitemOpen
  \bibfield  {author} {\bibinfo {author} {\bibfnamefont {J.~W.}\ \bibnamefont
  {Furness}}, \bibinfo {author} {\bibfnamefont {U.}~\bibnamefont
  {Ekstr{\"o}m}}, \bibinfo {author} {\bibfnamefont {T.}~\bibnamefont
  {Helgaker}}, \ and\ \bibinfo {author} {\bibfnamefont {A.~M.}\ \bibnamefont
  {Teale}},\ }\href@noop {} {\bibfield  {journal} {\bibinfo  {journal} {Mol.
  Phys.}\ }\textbf {\bibinfo {volume} {114}},\ \bibinfo {pages} {1415}
  (\bibinfo {year} {2016})}\BibitemShut {NoStop}%
\end{thebibliography}

\newcommand{\noopsort}[1]{} \newcommand{\printfirst}[2]{#1}
  \newcommand{\singleletter}[1]{#1} \newcommand{\switchargs}[2]{#2#1}
%

\end{document}